\documentclass[11pt, letterpaper]{article}
\usepackage{times}

\usepackage{amsmath,amsfonts,amssymb,amsthm}
\usepackage[colorlinks=true,linkcolor=red,citecolor=blue,urlcolor=red]{hyperref}
\usepackage{framed}
\usepackage{verbatim}
\usepackage[T1]{fontenc}
\usepackage{paralist}
\usepackage{stmaryrd}

\usepackage{mathptmx}
\usepackage[text={6.5in,9in}]{geometry}
\usepackage{enumerate}

\newenvironment{reminder}[1]{\medskip
\noindent {\bf Reminder of #1.\  }\em}{\smallskip}

\newtheorem{theorem}{Theorem}[section]
\newtheorem{proposition}{Proposition}[section]
\newtheorem{lemma}{Lemma}[section]
\newtheorem{problem}{Problem}[section]

\newenvironment{proofof}[1]{\medskip
\noindent {\bf Proof of #1.  }}{\hfill$\blacktriangleleft$
\medskip}

\def\eps{\varepsilon}
\def\poly{\text{poly}}
\def \Z {{\mathbb Z}}
\def \R {{\mathbb R}}

\def\polylog{\operatorname{polylog}}

\newcommand{\IGNORE}[1]{}
\begin{document}

\title{An Illuminating Algorithm for the Light Bulb Problem}

\date{}
\author{Josh Alman\footnote{MIT CSAIL. \texttt{jalman@mit.edu}. Much of this work was done while the author was working at IBM Research Almaden.}}

\maketitle

\begin{abstract}
The Light Bulb Problem is one of the most basic problems in data analysis. One is given as input $n$ vectors in $\{-1,1\}^d$, which are all independently and uniformly random, except for a planted pair of vectors with inner product at least $\rho \cdot d$ for some constant $\rho > 0$. The task is to find the planted pair.
The most straightforward algorithm leads to a runtime of $\Omega(n^2)$. Algorithms based on techniques like Locality-Sensitive Hashing achieve runtimes of $n^{2 - O(\rho)}$; as $\rho$ gets small, these approach quadratic.

Building on prior work, we give a new algorithm for this problem which runs in time $O(n^{1.582} + nd)$, \emph{regardless} of how small $\rho$ is. This matches the best known runtime due to Karppa et al. Our algorithm combines techniques from previous work on the Light Bulb Problem with the so-called `polynomial method in algorithm design,' and has a simpler analysis than previous work. Our algorithm is also easily derandomized, leading to a  deterministic algorithm for the Light Bulb Problem with the same runtime of $O(n^{1.582} + nd)$, improving previous results.
\end{abstract}

\thispagestyle{empty}
\newpage
\setcounter{page}{1}

\section{Introduction}

In this paper, we study the problem of finding correlated vectors. Finding correlations is one of the most basic problems in data analysis. In many experiments, one gathers data about a number of different variables, and then one would like to determine which variables are correlated. By forming the vector of data points for each variable, this amounts to finding which pairs of vectors are correlated.

The most basic formalization of this problem is the so-called Light Bulb Problem, introduced by L.~Valiant in 1988 \cite{lightbulb}:

\begin{problem}[Light Bulb Problem]
We are given as input a set $S$ of $n$ vectors from $\{-1,1\}^d$, which are all independently and uniformly random except for two planted vectors (the correlated pair) which have inner product at least $\rho \cdot d$ for some $0 < \rho \leq 1$. The goal is to find the correlated pair.
\end{problem}

The dimension $d$ of the vectors is called the \emph{sample complexity} of the problem, since, in our data analysis application, it corresponds to the number of data points which must be gathered about the variables in order to determine which are correlated. When $d$ is too small, then the problem is information-theoretically impossible. For instance, if $d < \log(n-2)$, then by the pigeonhole principle, two of the random vectors must be \emph{equal} to each other, and there is no way to distinguish them from the planted correlated pair we are trying to find. By standard concentration inequalities, there is a constant $c>1$ such that, whenever $d \geq c \log n$, the correlated pair is the closest pair of vectors with high probability. We would like to design algorithms for this $d = O(\log n)$ regime, so that we can find correlated pairs without increasing the sample complexity above the information-theoretic requirement. 

A na{\"i}ve approach to the Light Bulb Problem is to compute the inner product of each pair, which takes $\Omega(n^2)$ time. However, in many applications, $n$ is quite large, and quadratic time is infeasible. For one example, in genome-wide association studies, scientists have gathered data on millions of genetic markers, and determining which of these are correlated is key to understanding their interactions in different biological mechanisms \cite{genome1, genome2}.

It is not hard to see that the Light Bulb Problem is a special case of the \emph{$(1+\eps)-$approximate Hamming nearest neighbor problem}. Using Indyk and Motwani's famous Locality-Sensitive Hashing framework~\cite{lsh}, one can solve the Light Bulb Problem in time $n^{2 - O(\rho)}$. For constant $\rho>0$, this gives a truly subquadratic runtime, but the runtime become quadratic as $\rho \to 0$. This is undesirable, as in many data analysis applications, as well as in applications to other areas like learning theory, we would like to quickly detect weak correlations with small $\rho$. Later work \cite{lb1,lb2} improved the constants in the $O(\rho)$ term, but still had the same asymptotic dependence on $\rho$ in the exponent.

In a breakthrough result, G.~Valiant \cite{greg} gave an algorithm solving the Light Bulb Problem in time $O(n^{(5-\omega)/(4-\omega) + \eps} + nd) < O(n^{1.615} + nd)$, where $\omega < 2.373$ is the matrix multiplication constant, for \emph{any} constant $\rho > 0$, no matter how small. Thereafter, Karppa et al. \cite{rand} gave an improved algorithm with a runtime of $O(n^{2\omega/3 + \eps} + nd) < O(n^{1.582} + nd)$. Both of these algorithms work when the sample complexity $d$ matches, up to a constant, the information-theoretically necessary $d = \Theta(\log n)$.

In this paper, we give an algorithm with a simple analysis which matches the best known runtime and sample complexity.
\begin{theorem} \label{lightbulb}
For every $\eps, \rho > 0$, there is a $\kappa > 1$ such that the Light Bulb Problem for correlation $\rho$ can be solved in randomized time $O(n^{2\omega/3 + \eps})$ whenever $d = \kappa \log n$ with polynomially low error.
\end{theorem}
By leveraging our simpler analysis, we also give algorithms for several natural extensions and generalizations of the Light Bulb Problem, which are sometimes much faster than the algorithms from previous work.

\subsection{Algorithm Overview}
Our algorithm combines techniques from past work on subquadratic algorithms for the Light Bulb Problem \cite{greg, rand, det} with techniques for batch nearest neighbor algorithms using the \emph{`polynomial method in algorithm design'} \cite{aw, acw}. Our algorithm begins in a way common to both methods: we partition $S$ into $m = n^{2/3}$ groups $S = S_1 \cup \cdots \cup S_m$. Our goal is then to simultaneously check, for each pair $(S_i, S_j)$ of groups, whether there is a correlated pair of vectors in $S_i \times S_j$. We will do this by, for each group $S_i$, constructing two vectors $A_i, B_i \in \R^{t}$, for $t \approx n^{2/3}$, such that the inner product $\langle A_i, B_j \rangle$ is large if and only if there is a correlated pair of vectors in $S_i \times S_j$. By using fast matrix multiplication, we can quickly compute all of these inner products and find the correlated pair.

We differ from past work in how the $A_i$ and $B_i$ vectors are constructed. In \cite{greg, rand}, a sophisticated random sampling technique is used, including an involved probabilistic analysis to keep $t$ low, and in \cite{det}, that technique is derandomized. We instead use the polynomial method: we first design a polynomial (see (\ref{defC}) below) which, it is not hard to show, has the desired properties. We then convert it into $A_i, B_i$ vectors by dividing it up into monomials. By designing a polynomial whose degree is not too high, we get that the resulting number of monomials, and hence $t$, is also not too high. Our proof of correctness is straightforward, and it almost entirely avoids arguments about tail distributions of sums of dependent variables, including a multitude of calculations and casework, which are prevalent in the past work.

That said, our algorithm can be seen as the `best of both worlds': past work on the polynomial method has focused on designing subquadratic time algorithms, but not on optimizing \emph{how} subquadratic the runtime is. We use ideas from past work on the Light Bulb Problem (our overall approach to the problem comes from \cite{greg}, and the last paragraph in the proof of Theorem~\ref{lightbulb} uses a clever trick of \cite{rand}) in order to optimize our runtime here.

\subsection{Deterministic Light Bulb Problem}

In some cases, one would like a deterministic algorithm which is guaranteed to find correlations in subquadratic time. Since our polynomial construction and evaluation process is entirely deterministic, we can get such an algorithm easily. However, as the inputs to the Light Bulb Problem come from a random distribution, we need to be careful about what a deterministic algorithm means in this setting. For instance, there is a small chance that a random pair of vectors will be just as correlated as the correlated pair, in which case a deterministic algorithm has no hope of finding the true correlated pair.

\vspace{2mm}
\noindent{\textbf{Almost All Instances}}

One option is to design an algorithm which correctly solves \emph{almost all instances}. This is the notion which was introduced and used in the past work by Karppa et al.~\cite{det} on deterministic algorithms for the Light Bulb Problem. We say that an algorithm is correct on almost all instances if the probability of drawing an instance where the algorithm fails is $1/\poly(n)$. For this notion, we match the runtime of the best randomized algorithm:

\begin{theorem} \label{lightbulbalmostall}
For every $\eps, \rho > 0$, there is a $\kappa > 0$ such that the Light Bulb Problem for correlation $\rho$ can be solved in deterministic time $O(n^{2\omega/3 + \eps})$ on almost all instances whenever $d = \kappa \log n$.
\end{theorem}

Our runtime of $O(n^{2\omega/3 + \eps}) \leq O(n^{1.582})$ is faster than the runtime of Karppa et al.~\cite{det}, which is at best $O(n^{1.996})$. Our algorithm also uses more straightforward and elementary techniques. The original algorithm of Karppa et al. relies heavily on random sampling. In order to derandomize this, Karppa et al. use heavy-duty techniques including constructing `correlation amplifiers' using the explicit expander graphs of Reingold, Vadhan, and Wigderson \cite{expander}. We avoid any such complications, since our algorithm replaces random sampling with a deterministic polynomial construction. In fact, one can view our polynomials in (\ref{defC}), below, as a smaller, elementary construction of their notion of a correlation amplifier.

Our algorithm for Theorem \ref{lightbulbalmostall} uses two ideas to derandomize the algorithm for Theorem \ref{lightbulb} without changing the runtime. First, we use a standard technique in derandomization: by examining the proof of correctness of Theorem \ref{lightbulb}, we will see that the random bits it uses only need to be \emph{pairwise independent}, rather than fully independent, which means only $O(\polylog n)$ independent random bits are needed for the algorithm to succeed with high probability. Second, in the proof of Theorem \ref{lightbulbalmostall}, we use the fact that, with the exception of the correlated pair of vectors, the vectors in the input set $S$ of the Light Bulb Problem \emph{are} random vectors, and we can use them as the source of random bits we need. This technique of using the input as a source of randomness has been used in a number of past derandomization results; see eg. \cite{inputrand1, inputrand2}.

\vspace{2mm}
\noindent{\textbf{Promise that random vectors aren't too correlated}}

Although the algorithm of \cite{det} is presented as working on almost all instances, it implicitly works in a stronger regime. It solves a promise version of the Light Bulb Problem, which we introduce here, in which we are guaranteed that no pair of random vectors is too correlated:

\begin{problem}[Promise Light Bulb Problem with parameter $w$]
We are given as input a set $S$ of $n$ vectors from $\{-1,1\}^d$, where two of the vectors (the correlated pair) have inner product at least $\rho \cdot d$ for some $0 < \rho \leq 1$, and every other pair of vectors has inner product at most $w \sqrt{d \log n}$. The goal is to find the correlated pair.
\end{problem}

To emphasize: the inputs to the Promise Light Bulb Problem are not necessarily chosen randomly; they can be chosen adversarially as long as they satisfy the guarantee.

By a Chernoff bound, a random instance of the Light Bulb Problem will also satisfy this guarantee with probability $1 - 1/\poly(n)$, for a sufficiently large constant $w$. Hence, deterministically solving the Promise Light Bulb Problem is sufficient to solve the Light Bulb Problem on almost all instances, and this is the approach that \cite{det} takes. One benefit of the Promise Light Bulb Problem is that it doesn't let us use the `artificial' trick of using vectors from a randomly chosen input as the source of randomness. Without using that trick, we can nonetheless solve the Promise Light Bulb Problem deterministically, with running time $O(n^{4\omega/5 + \eps}) \leq O(n^{1.8983})$:

\begin{theorem} \label{lightbulbdet}
There is a constant $w > 0$ such that, for every $\eps, \rho > 0$, there is a $\kappa > 0$ such that the Promise Light Bulb Problem with parameter $w$ for correlation $\rho$ can be solved in deterministic time $O(n^{4\omega/5 + \eps})$ whenever $d = \kappa \log n$.
\end{theorem}

While this algorithm is slower than our aforementioned algorithms, it is nonetheless still faster than the previous best deterministic runtime \cite{det} of $O(n^{1.996})$, and it follows without much more work from our deterministic polynomial construction.

\subsection{Generality of the Light Bulb Problem}

As the Light Bulb Problem is so basic, a number of other important problems can be reduced to it as well. Here we give a couple of examples from prior work.

\noindent{\textbf{Correlations on the Euclidean Sphere}}

In all the above, we have been discussing finding correlated vectors from the domain $\{-1,1\}^d$. What if we are more generally interested in finding correlated vectors from the $d$-dimensional Euclidean sphere\footnote{The $d$-dimensional Eucliedan sphere is the set of points $x \in \R^d$ such that $x_1^2 + \cdots + x_d^2 = 1$.}? There is a randomized hashing algorithm by Charikar \cite{rounding} that `rounds' the Euclidean sphere to $\{-1,1\}^d$ in such a way that all of our algorithms above will still work, with $\rho$ only decreasing by a constant factor. The hash function simply picks a uniformly random hyperplane through the origin, and outputs $1$ or $-1$ depending on which side of the hyperplane a point lies on. Since our runtime for the Light Bulb Problem in Theorem~\ref{lightbulb} does not change when $\rho$ changes by a constant factor, we can thus achieve the same guarantees for the Light Bulb Problem on the Euclidean sphere.

\vspace{2mm}
\noindent{\textbf{Learning Sparse Parities with Noise and More}}

L.~Valiant \cite{lightbulb} first introduced the Light Bulb Problem as a basic example of a correlated learning problem. More generally, the Light Bulb Problem can be seen as a special case of several different problems in learning theory, including learning sparse parities with noise, learning sparse juntas with or without noise, and learning sparse DNFs. Surprisingly, Feldman et al.~\cite{vitaly} showed that all these more general learning problems can be reduced to the Light Bulb Problem as well, and the fastest known algorithms for them come from applying this reduction followed by the best Light Bulb Problem algorithms. Hence, our algorithm gives a new, simpler algorithm matching the best known runtimes for these problems as well. We refer to \cite[Appendix A]{greg} for a more detailed discussion of these reductions.

\section{Preliminaries} \label{prelims}

We assume familiarity with basic facts about combinatorics and probability, and in particular, the union bound, Chernoff bound, and Chebyshev inequality.
For an integer $d \geq 0$, we write $[d] := \{ 1,2,\ldots,d\}$. For a vector $x \in \{-1,1\}^d$, we will write $x_i$ to denote the $i$th entry of $x$ for any $i \in [d]$, and $x_M := \prod_{i \in M} x_i$ for any $M \subseteq [d]$.

\vspace{2mm}
\noindent{\textbf{Polynomial Multilinearization}}

For a multivariate polynomial $p : \R^d \to \R$, its \emph{multilinearization} is the polynomial $\bar{p} : \R^d \to \R$ which one gets when one expands $p$ into a sum of monomials, and then for each monomial, and each variable in that monomial, one reduces the exponent of that variable mod 2 to either 0 or 1. For instance, if $p(x_1, x_2, x_3) = x_1 x_2^5 x_3^2 + 3x_2^2$, then $\bar{p}(x_1, x_2, x_3) = x_1 x_2 + 3$. Notice that $x_i^2 = 1$ whenever $x_i \in \{-1,1\}$, and so for any $p$, and any $x \in \{-1,1\}^d$, we always have that $p(x) = \bar{p}(x)$. The number of multilinear monomials on $d$ variables of degree exactly $r$ is $\binom{d}{r}$. Hence, if $p$ has degree $r$, then the number of monomials in $\bar{p}$ is at most $\sum_{i=0}^r \binom{d}{i}$.

 We will use the two bounds on binomial coefficients to bound the number of monomials in $\bar{p}$. First, if $0 \leq k_1 \leq k_2 \leq n/2$, then $\binom{n}{k_1} \leq \binom{n}{k_2}$. Second, for any $1 \leq k \leq n$, Stirling's approximation shows that
\begin{align} \label{binbound} \binom{n}{k} \leq \frac{n^k}{k!} \leq \left( \frac{e \cdot n}{k} \right)^k.\end{align}

\vspace{2mm}
\noindent{\textbf{Matrix Multiplication Notation}}

Let $M(a,b)$ denote the runtime to compute the product of an $a \times b$ matrix with a $b \times a$ matrix, whose entries are integers of magnitude at most $2^{\polylog(ab)}$. For instance, $M(n,n) \leq O(n^{\omega})$ where $\omega \leq 2.373$ \cite{virginia, legall} is the matrix multiplication exponent. Since a $n \times n^{1+\eps} \times n$ matrix multiplication can be decomposed into $n^{\eps}$ different $n \times n \times n$ multiplications, we see that for $\eps \geq 0$, \begin{align}\label{MM} M(n, n^{1+\eps}) \leq O(n^{\omega + \eps}). \end{align}
\newpage
\section{Algorithm for the Light Bulb Problem}

In this section, we give our algorithm for the Light Bulb Problem, proving our main result, Theorem~\ref{lightbulb}. 

\begin{reminder}{Theorem~\ref{lightbulb}}
For every $\eps, \rho > 0$, there is a $\kappa > 0$ such that the Light Bulb Problem for correlation $\rho$ can be solved in randomized time $O(n^{2\omega/3 + \eps})$ whenever $d = \kappa \log n$ with polynomially low error.
\end{reminder}

For two constants $\gamma, k > 0$ to be determined, we will pick $\kappa = \gamma k^2 / \rho^2$. Let $S \subseteq \{-1,1\}^d$ be the set of input vectors, and let $x', y' \in S$ denote the correlated pair which we are trying to find. For distinct $x,y \in S$ other than the correlated pair, the inner product $\langle x, y\rangle$ is a sum of $d$ uniform independent $\{-1,1\}$ values. Let $v := \gamma (k / \delta) \log n$. By a Chernoff bound, for large enough $\gamma$, we have $|\langle x,y \rangle| \leq v$ with probability at least $1-1/n^3$. Hence, by a union bound over all pairs of uncorrelated vectors, we have $|\langle x,y \rangle| \leq v$ for all such $x,y$ with probability at least $1-1/n$. We assume henceforth that this is the case. Meanwhile, $\langle x', y' \rangle \geq \rho d = kv$.


Arbitrarily partition $S$ into $m := n^{2/3}$ groups $S_1, \ldots, S_m$ of size $g := n/m = n^{1/3}$ each. We can compute the inner product between each pair of vectors which was assigned to the same group in time $O(m \cdot g^2 \cdot d) = \tilde{O}(n^{4/3})$, and if we find the correlated pair, we can return it and end the algorithm. Otherwise, we may assume the correlated vectors are in different groups, and we continue.

For each $x \in S$, our algorithm picks a value $a^{x} \in \{-1,1\}$ independently and uniformly at random. For a constant $\tau > 0$ to be determined, let $r = \lceil \log_k(\tau n^{1/3}) \rceil$, and define the polynomial $p : \R^d \to \R$ by $p(z_1, \ldots, z_d) = (z_1 + \cdots + z_d)^r$. Our goal is, for each $(i,j) \in [m]^2$, to compute the value $$C_{i,j} := \sum_{x \in S_i} \sum_{y \in S_j} a^{x} \cdot a^{y} \cdot p(x_1 y_1, \ldots, x_d y_d).$$

\paragraph{Solving the problem using $C_{i,j}$} Let us first explain why we are interested in computing $C_{i,j}$. Denote $p(x,y) := p(x_1 y_1, \ldots, x_d y_d)$. Intuitively, $p(x,y)$ is computing an \emph{amplification} of $\langle x,y \rangle$. $C_{i,j}$ is then summing these amplified inner products for all pairs $(x,y) \in S_i \times S_j$. We will pick our parameters so that the amplified inner product of the correlated pair is large enough to stand out from the sums of inner products of random pairs.

Let us be more precise. Recall that for uncorrelated $x,y$ we have $|\langle x,y \rangle| \leq v$, and hence $|p(x,y)| \leq v^r$. Similarly, we have $|p(x',y')| \geq (kv)^r \geq \tau n^{1/3} v^r$. For $x,y \in S$, define $a^{(x,y)} := a^x \cdot a^y$. Notice that, for $i \neq j$, $C_{i,j} = \sum_{x \in S_i, y \in S_j} a^{(x,y)} p(\langle x,y \rangle)$, where the $a^{(x,y)}$ are \emph{pairwise independent} random $\{-1,1\}$ values.

We will now analyze the random variable $C_{i,j}$ where we think of the vectors in $S$ as fixed, and only the values $a^x$ as random.

Consider first when the correlated pair are not in $S_i$ and $S_j$. Then, $C_{i,j}$ has mean $0$, and (since variance is additive for pairwise independent variables) $C_{i,j}$ has variance at most $|S_i| \cdot |S_j| \cdot \max_{x \in S_i, y \in S_j} |p(\langle x, y \rangle)|^2 \leq n^{2/3} \cdot v^{2r}$. For sufficiently large constant $\tau$, by the Chebyshev inequality, we have that $|C_{i,j}| \leq \tau n^{1/3} v^r / 3$ with probability at least $3/4$. Let $\theta = \tau n^{1/3} v^r / 3$, so $|C_{i,j}| \leq \theta$ with probability at least $3/4$.

Meanwhile, if $x' \in S_i$ and $y' \in S_j$, then $C_{i,j}$ is the sum of $a^{(x',y')} p(\langle x',y' \rangle)$ and a variable $C'$ distributed as $C_{i,j}$ was in the previous paragraph. Hence, since $|p(\langle x',y' \rangle)| \geq \tau n^{1/3} v^r = 3 \theta$, and $|C'| \leq \theta$ with probability at least $3/4$, we get by the triangle inequality that $|C_{i,j}| \geq 2 \theta$ with probability at least $3/4$.

Hence, if we repeat the process of selecting the $a^{x}$ values for each $x \in S$ independently at random $O(\log n)$ times, whichever pair $S_i, S_j$ has $|C_{i,j}| \geq 2 \theta$ most frequently will be the pair containing the correlated pair with polynomially low error, and then a brute force within this set of $O(n^{1/3})$ vectors can find the correlated pair in $\tilde{O}(n^{2/3})$ time. In all, by a union bound over all possible errors, this will succeed with polynomially low error.

\paragraph{Computing $C_{i,j}$} It remains to give the algorithm to compute $C_{i,j}$. Before doing this, we will rearrange the expression for $C_{i,j}$ into one which is easier to compute. Since we are only interested in the values of $p$ when its inputs are all in $\{-1,1\}$, we can replace $p$ with its multilinearization $\hat{p}$. Let $M_1, \ldots, M_t$ be an enumeration of all subsets of $[d]$ of size at most $r$, so $t = \sum_{i=0}^r \binom{d}{i}$. Then, there are coefficients $c_1, \ldots, c_t \in \Z$ such that $\hat{p}(x) = \sum_{s=1}^t c_s x_{M_s}$. Rearranging the order of summation, we see that we are trying to compute
\begin{align} \label{defC} C_{i,j} = \sum_{s=1}^t \sum_{x \in S_i} \sum_{y \in S_j} a^{x} \cdot a^{y} \cdot c_s \cdot x_{M_s} \cdot y_{M_s} = \sum_{s=1}^t \left[ c_s \cdot \left( \sum_{x \in S_i} a^{x}  \cdot  x_{M_s} \right) \cdot \left( \sum_{y \in S_j} a^{y}  \cdot  y_{M_s} \right) \right] .\end{align}
In order to compute $C_{i,j}$, we first need to compute the coefficients $c_s$. Notice that $c_s$ depends only on $|M_s|$ and $r$. We can thus derive a simple combinatorial expression for $c_s$, and hence compute all of the $c_s$ coefficients in $\poly(r) = \polylog(n)$ time. Alternatively, by starting with the polynomial $(z_1 + \cdots + z_d)$ and then repeatedly squaring then multilinearizing, we can easily compute all the coefficients in $O(t^2 \polylog(n))$ time; this slower approach is still fast enough for our purposes.

Define the matrices $A, B \in \Z^{m \times t}$ by $A_{i,s} = \sum_{x \in S_i} a^{x}  \cdot  x_{M_s} $ and $B_{i,s} = c_s \cdot A_{i,s}$. Notice from (\ref{defC}) that the matrix product $C := A B^T$ is exactly the matrix of the values $C_{i,j}$ we desire. A simple calculation (see Lemma~\ref{boundtufirst} below) shows that for any $\eps>0$, we can pick a sufficiently big constant $k>0$ such that $t = O(n^{2/3 + \eps})$. Since $m = O(n^{2/3})$, if we have the matrices $A,B$, then we can compute this matrix product in $M(n^{2/3}, n^{2/3 + \eps}) = O(n^{2\omega/3+ \eps})$ time, completing the algorithm.

Unfortunately, computing the entries of $A$ and $B$ naively would take $\Omega(m \cdot t \cdot g) = \Omega(n^{5/3})$ time, which is slower than we would like. We will instead use a clever trick due to Lovett~\cite{originaltrick}, which was first applied in this context by Karppa et al.~\cite{rand}: we will compute those entries using \emph{another} matrix multiplication. Let $N_1, \ldots, N_u$ be an enumeration of all subsets of $[d]$ of size at most $\lceil r/2 \rceil$. For each $i \in [m]$, define the matrices $L^i, \tilde{L^i} \in \Z^{u \times g}$ (whose columns are indexed by elements $x \in S_i$) by $L^i_{s,x} = x_{N_s}$ and $\tilde{L^i}_{s,x} = a^{x} \cdot x_{N_s}$. Then, compute the product $P^i := L^i \tilde{L^{i}}^T$. We can see that $P^i_{s,s'} = \sum_{x \in S_i} a^{x}  \cdot  x_{N_s \oplus N_{s'}}$, where $N_s \oplus N_{s'}$ is the symmetric difference of $N_s$ and $N_{s'}$. Since any set of size at most $r$ can be written as the symmetric difference of two sets of size at most $\lceil r/2 \rceil$, each desired entry $A_{i,s}$ can be found as an entry of the computed matrix $P^i$. Similar to our bound on $t$ from before (see Lemma~\ref{boundtufirst} below), we see that for big enough constant $k$, we have $u = O(n^{1/3 + \eps})$. Computing the entries of the $L^i$ matrices naively takes only $O(m \cdot u \cdot g \cdot r) = \tilde{O}(n \cdot u)= \tilde{O}(n^{4/3 + \eps})$ time, and then computing the products $P^i$ takes $O(m \cdot \max(u,g)^\omega) = O(n^{(2+ \omega)/3 + \eps})$ time; both of these are dominated by $O(n^{2\omega/3+ \eps})$. This completes the algorithm! Finally, we perform the computations mentioned above:

\begin{lemma} \label{boundtufirst}
For every $\eps > 0$, there is a $k > 0$ such that (with the same notation as in the proof of Theorem \ref{lightbulb} above) we can bound $t = O(n^{2/3 + \eps})$, and $u = O(n^{1/3 + \eps})$.
\end{lemma}

\begin{proof}
Recall that $d = O(k^2 \log(n))$, and $r = \log_k(O(n^{1/3}))$. Hence, by the bound (\ref{binbound}),
$$t \leq (r+1) \cdot \binom{d}{r} \leq (r+1) \cdot (ed/r)^r \leq O(k^2 \log(k))^{\log_k(O(n^{1/3}))} = n^{2/3 + O(\log\log(k)/\log(k))}.$$
For any $\eps>0$ we can thus pick a sufficiently large $k$ so that $t \leq O(n^{2/3 + \eps})$. We can similarly bound $\binom{d}{r/2} \leq O(n^{1/3 + \eps})$ which implies our desired bound on $u$.
\end{proof}

\section{Deterministic Algorithms}

We now present our two deterministic algorithms for the Light Bulb Problem. Each is a slight variation on the algorithm from the previous section.

\begin{reminder}{Theorem~\ref{lightbulbalmostall}}
For every $\eps, \rho > 0$, there is a $\kappa > 0$ such that the Light Bulb Problem for correlation $\rho$ can be solved in deterministic time $O(n^{2\omega/3 + \eps})$ on almost all instances whenever $d = \kappa \log n$.
\end{reminder}

\begin{proof}
The only randomness used by our algorithm for Theorem~\ref{lightbulb} was our choice of an independently and uniformly random $a^x \in \{-1,1\}$ for each $x \in S$. Since this requires $\Theta(n)$ random bits, and we repeat the entire algorithm $\Theta(\log n)$ times to get our desired correctness guarantee, the total number of random bits used is $\Theta(n \log n)$.

However, the only property of the $a^x$ variables which we use in the proof of correctness is that they are \emph{pairwise}-independent. By standard constructions\footnote{For one example, to generate $2^\ell - 1$ pairwise-independent bits, pick only $\ell$ bits $b_1, \ldots, b_\ell \in \{-1,1\}$ independently and uniformly at random, and then output, for each $I \subseteq [\ell]$, the product $\prod_{i \in I} b_i$.}, only $O(\log n)$ independent random bits are needed to generate $n$ pairwise-independent random bits. Thus, our entire algorithm actually only needs $O(\log^2 n)$ independent random bits.

Our entirely deterministic algorithm then proceeds as follows. Pick the same $\kappa$ as in Theorem~\ref{lightbulb}. Let $S \subseteq \{-1,1\}^{d}$ be the input vectors. Arbitrarily pick a subset $S' \subseteq S$ of $|S'| = \Theta(\log n)$ of the input vectors, and let $R = S \setminus S'$ be the remaining vectors. 

We begin by testing via brute-force whether either vector of the correlated pair is in $S'$. This can be done in $O(|S'| \cdot |S| \cdot d) = O(n \log^2(n))$ time. If we find the correlated pair (a pair with inner product at least $\rho \cdot d$), then we output it, and otherwise, we can assume that the vectors in $S'$ are all uniformly random vectors from $\{-1,1\}^d$. In other words, we can use them as $d \cdot |S'| = \Theta(\log^2 n)$ independent uniformly random bits. We thus use them as the required randomness to run the algorithm from Theorem~\ref{lightbulb} on input vectors $R$. That algorithm has polynomially low error, which implies the desired correctness guarantee.
\end{proof}

\begin{reminder}{Theorem \ref{lightbulbdet}}
There is a constant $w > 0$ such that, for every $\eps, \rho > 0$, there is a $\kappa > 0$ such that the Promise Light Bulb Problem with parameter $w$ for correlation $\rho$ can be solved in deterministic time $O(n^{4\omega/5 + \eps})$ whenever $d = \kappa \log n$.
\end{reminder}

\begin{proof}
The guarantee of the Promise Light Bulb Problem is that, when we pick a sufficiently large $w$, the uncorrelated vectors have as small inner product as we assumed they did in the first paragraph in the proof of Theorem \ref{lightbulb}. In other words, there is a quantity $v$ such that $|\langle x,y \rangle| \leq v$ for all $x,y \in S$ other than the correlated pair, and moreover, $\langle x', y' \rangle \geq kv$ for a constant $k>0$ with $k \to \infty$ as $w \to \infty$.

The algorithm is then almost identical to Theorem \ref{lightbulb}, except we need to remove the only use of randomness: the randomness used to pick the $a^x$ values. To do this, we will simply pick $a^x = 1$ for all $x$.

In order to guarantee the correctness of our algorithm, we must now change the parameters slightly. Instead of partitioning the input into $m = n^{2/3}$ groups of size $g = n^{1/3}$, we will instead partition into $m = n^{4/5}$ groups of size $g = n^{1/5}$. Similarly, instead of picking $r$ (the exponent in the polynomial $p$) to be $\log_k(O(n^{1/3}))$, we will pick $r = \log_k(3 n^{2/5})$, so that $p(x',y') \geq (kv)^r = 3 n^{2/5} v^r$.

With these choices, for any $i$ and $j$ such that the correlated pair are not in $S_i$ and $S_j$, we have $|C_{i,j}| \leq |S_i| \cdot |S_j| \cdot n^{2/5} = n^{2/5} v$, whereas if $x' \in S_i$ and $y' \in S_j$ then by the triangle inequality, $|C_{i,j}| \geq p(x',y') -|S_i| \cdot |S_j| \cdot n^{2/5} \geq 2 n^{2/5} v^r$. Hence, the correlated pair must be in whichever $S_i$ and $S_j$ with $i \neq j$ has the largest $|C_{i,j}|$.

The algorithm to compute the $C_{i,j}$ values is identical to that of Theorem \ref{lightbulb}. We now get that $t = \sum_{i=0}^{r} \binom{d}{i} \leq O(n^{4/5 + \eps})$ and similarly, $u \leq O(n^{2/5 + \eps})$, which leads to a final runtime of $O(n^{4 \omega / 5 + \eps})$, as desired.
\end{proof}

\section{Conclusion}

\noindent{\textbf{Faster Algorithms?}}

In this paper, we give an algorithm for the Light Bulb Problem and some variants. A natural question remains: can one improve the $O(n^{2\omega/3})$ runtime? It seems like substantially new techniques might be necessary. We currently reduce the problem to a $n^{2/3} \times n^{2/3} \times n^{2/3}$ matrix multiplication; with a further reduction in the dimensions, even using the cubic matrix multiplication algorithm would give a subquadratic algorithm for the Light Bulb Problem. This would be surprising, since recent progress on the problem has relied heavily on fast matrix multiplication.

It should nonetheless be noted that, despite using fast matrix multiplication, the algorithms in this paper can be quite practical. For instance, using Strassen's original algorithm \cite{strassen}, which is frequently used in practice, gives $\omega \approx 2.81$, and hence a subquadratic runtime for the Light Bulb Problem of about $O(n^{2\omega/3}) \leq O(n^{1.874})$.

\noindent{\textbf{Finding General Correlations}}

\noindent{Past work on the Light Bulb Problem has also approached a more general problem of finding correlations:}

\begin{problem}[Finding Correlations]\label{findcorr}We are given as input two sets $X,Y \subseteq \{-1,1\}^d$ of $n$ vectors each, with the promise that for at most $q$ pairs of $x,y \in X \times Y$, we have $|\langle x,y \rangle| \geq \rho d$ ($x$ and $y$ are \emph{correlated}), and for all other pairs of $x,y \in X \times Y$, we have $|\langle x,y \rangle| \leq \tau d$ ($x$ and $y$ are \emph{uncorrelated}) for some constants $0 < \tau < \rho \leq 1$. Our goal is to find all $q$ of the correlated pairs of vectors.
\end{problem}

Again, as $\rho \to 0$, hashing techniques give runtimes which approach quadratic. However, if $\tau$ is also comparatively small (say, there is a constant $\sigma > 1$ such that $\rho / \tau \geq \sigma$), we might hope to achieve a truly subquadratic runtime, no matter how small $\rho$ becomes. With only a slight modification of our algorithm for Theorem~\ref{lightbulbdet}, we can achieve this:

\begin{proposition} \label{subquad}
For all constants $\eta, c > 0$, and $\sigma > 1$, there exists a constant $\eps > 0$ such that Finding Correlations can be solved in $O(n^{2-\eps})$ deterministic time when $\rho/\tau \geq \sigma$, $q \leq n^{2-\eta}$ and $d = c \log(n)$.
\end{proposition}

Proposition \ref{subquad} is somewhat weaker than the results from past work \cite{greg, det, rand}, which only require that $\log(1/\tau)/\log(1/\rho)$ be bounded below by a constant. However, our algorithm benefits from the same simplicity as our Light Bulb Problem algorithms, and it is also deterministic (in the usual sense -- there is no distribution on inputs the Finding Correlations problem). Like before, only \cite{det} gives a deterministic algorithm, and it involves the same aforementioned heavy-duty techniques which we avoid. We omit the details of this algorithm here for clarity of exposition, as the algorithm is almost identical to that of Theorem~\ref{lightbulbdet}.

\vspace{2mm}
\noindent{\textbf{Acknowledgments}}
The author would like to thank Vitaly Feldman, Michael P. Kim, Virginia Vassilevska Williams, Ryan Williams, and anonymous reviewers for their comments on an earlier draft.

\bibliographystyle{alpha}
\bibliography{papers}

\end{document}